\documentclass{article}

\usepackage[tmargin=1in,bmargin=1in,lmargin=1.25in,rmargin=1.25in]{geometry}
\usepackage{graphicx}
\usepackage{bm}
\usepackage{amsfonts}
\usepackage{amsmath}
\usepackage{amssymb}
\usepackage{booktabs}
\usepackage{amsfonts}
\usepackage{sectsty}
\usepackage[compact]{titlesec}
\usepackage{color}
\usepackage[usenames,dvipsnames,svgnames,table]{xcolor}
\usepackage{authblk}

\usepackage{mathptmx}

\newcommand{\R}{\mathbb{R}}
\newcommand{\B}{\mathcal{B}}
\newcommand{\X}{\bm{X}}
\newcommand{\e}{\varepsilon}

\usepackage{times}
\usepackage{bm}
\usepackage{bbm}
\usepackage[round]{natbib}
\usepackage{hyperref}

\usepackage[plain,noend]{algorithm2e}

\usepackage{amsthm}
\newtheorem{definition}{Definition}
\newtheorem{proposition}{Proposition}
\newtheorem{lemma}{Lemma}
\newtheorem{example}{Example}

\newtheorem{remark}{Remark}

\makeatletter
\renewcommand{\algocf@captiontext}[2]{#1\algocf@typo. \AlCapFnt{}#2} 
\def\@algocf@capt@plain{top}
\renewcommand{\algocf@makecaption}[2]{%
  \addtolength{\hsize}{\algomargin}%
  \sbox\@tempboxa{\algocf@captiontext{#1}{#2}}%
  \ifdim\wd\@tempboxa >\hsize
    \hskip .5\algomargin%
    \parbox[t]{\hsize}{\algocf@captiontext{#1}{#2}}
  \else%
    \global\@minipagefalse%
    \hbox to\hsize{\box\@tempboxa}
  \fi%
  \addtolength{\hsize}{-\algomargin}%
}
\makeatother

\title{\Large \bfseries Kinetic energy choice in Hamiltonian/hybrid Monte
Carlo}
\author[1]{Samuel Livingstone\footnote{samuel.livingstone@bristol.ac.uk}}
\author[1]{Michael F. Faulkner}
\author[2]{Gareth O. Roberts}

\affil[1]{\normalsize\itshape School of Mathematics, University of Bristol, University Walk,\\ Bristol BS8 1TW, U.K.}
\affil[2]{\normalsize\itshape Department of Statistics, University of Warwick, Coventry CV4 7AL, U.K.}

\date{\today}

\begin{document}

\maketitle

\markboth{}{}




\maketitle

\begin{abstract}
We consider how different choices of kinetic energy in Hamiltonian Monte Carlo affect algorithm performance.  To this end, we introduce two quantities which can be easily evaluated, the composite gradient and the implicit noise.  Results are established on integrator stability and geometric convergence, and we show that choices of kinetic energy that result in heavy-tailed momentum distributions can exhibit an undesirable negligible moves property, which we define.  A general efficiency-robustness trade off is outlined, and implementations which rely on approximate gradients are also discussed.  Two numerical studies illustrate our theoretical findings, showing that the standard choice which results in a Gaussian momentum distribution is not always optimal in terms of either robustness or efficiency.
\end{abstract}

\section{Introduction}

\label{intro} 

Hamiltonian Monte Carlo is a Markov chain Monte Carlo method which is now both widely used in Bayesian inference, and increasingly studied and developed.  The idea of the approach is to use the deterministic measure-preserving dynamics of Hamiltonian flow to promote fast exploration of a parameter space of interest.  To achieve this the space is augmented with a momentum variable, and the Markov chain evolves by switching between re-sampling this momentum and solving Hamilton's equations for a prescribed amount of time.  Typically the equations cannot be solved exactly and so a time-reversible and volume-preserving numerical integrator is used, with discretisation errors controlled for using a Metropolis step.  Comprehensive reviews are given in \cite{neal2011mcmc} and \cite{betancourt2017geometric}.

The free variables in Hamiltonian Monte Carlo are the time for which Hamilton's equations should be solved between momentum refreshments, the choice of numerical integrator and step-size, and the choice of distribution for the momentum.  Typically the St\"ormer--Verlet or leapfrog numerical integrator is used, which gives a reasonable balance between energy preservation and computational cost.  Guidelines for step-size choice are given in \cite{beskos2013optimal}, which is typically tuned to reach a 65-80\% Metropolis acceptance rate.  Heuristics exist for the integration time \citep{hoffman2014no}, and there is some justification for a stochastic choice given in \cite{bou2015randomized}.

This paper is concerned with the choice of momentum distribution $\nu(\cdot)$, addressing a question which has been raised previously \citep{barthelme2011discussion,stephens2011discussion}.  We will assume throughout that $\nu(\cdot)$ possesses a density which is proportional to $\exp\{-K(p)\}$ for some $K(p)$, which we call the kinetic energy of the system.  The standard choice is $K(p) = p^T p/2$, with the resulting $\nu(\cdot)$ a Gaussian distribution.  The general requirements are simply that $K(p)$ is differentiable, symmetric about zero and that $\nu(\cdot)$ can be sampled from.  Alternative options to the Gaussian have recently been suggested \citep{lu2016relativistic,zhang2016towards}.  Here we consider how such a choice affects the algorithm as a whole, in terms of stability and convergence, and develop guidelines for practitioners.

Our key findings are that the robustness and efficiency of the method can to some degree be controlled through the quantity $\nabla K\circ \nabla U(x)$, notation defined below, which we term the composite gradient. In particular, we propose that balancing the tails of the kinetic energy with those of the potential to make this approximately linear in $x$ should give good performance, and that no faster than linear growth is necessary for algorithm stability.  When $\pi(\cdot)$ is very light-tailed, this can be done by choosing $\nu(\cdot)$ to be heavier-tailed. There are, however, serious disadvantages to choosing any heavier tails than those of a Laplace distribution, which can result in the sampler moving very slowly in certain regions of the space.  We introduce a negligible moves property for Markov chains to properly characterize this behaviour.  We also find that in practice considering the distribution of $\nabla K(p)$, which we term the implicit noise, is important, since this governs the behaviour of the sampler in regions where $\|\nabla U(x)\|$ is small.  We suggest various choices for controlling these two quantities, and support our findings with a numerical study.  We also comment on how changing the kinetic energy can affect behaviour when $\nabla U(x)$ is approximated, either through subsampling or within an exact-approximate Monte Carlo scheme.


\subsection{Setup, notation and assumptions}

Throughout we denote the distribution from which expectations are desired by $\pi(\cdot)$.  We assume that $\pi(\cdot)$ possesses a density which is proportional to $\exp\{-U(x)\}$, for some $C^1(\R^d)$ potential $U:\R^d \to [0,\infty)$, where $x \in \R^d$.  The Hamiltonian is formed as $H(x,p) = U(x) + K(p)$, where $p\in \R^d$ and $K:\R^d \to [0,\infty)$ is a continuously differentiable kinetic energy, and the corresponding Hamiltonian dynamics are given by
\begin{equation} \label{eqn:hdynamics}
\dot{x} = \nabla K(p), ~~~~ \dot{p} = -\nabla U(x).
\end{equation}
We restrict ourselves to separable Hamiltonians, meaning the distribution for $p$ does not depend on $x$, as in this case general purpose and well-studied explicit numerical integrators are available \citep{leimkuhler2004simulating}.  We also focus on the leapfrog numerical integrator \citep{neal2011mcmc} and a choice of integration time $T = L\e$, where $\e$ is the integrator step-size and $L$ is the number of leapfrog steps.  We assume that $L$ is drawn from some distribution $\Psi(\cdot)$, which is independent of the current $x$ and $p$, and denote the result of solving \eqref{eqn:hdynamics} using this integrator for $T$ units of time given the starting position and momentum $(x,p)$ as $\tilde{\varphi}_T(x,p)$.  With these tools, the Hamiltonian Monte Carlo method is described as Algorithm \ref{al1} below.

\begin{algorithm}[!h]
\caption{Hamiltonian Monte Carlo} \label{al1}
\vspace*{-12pt}
\begin{tabbing}
   \enspace Require $x^0, \e$\\
   \enspace For $i=1$ to $i=n$ \\
   \qquad Draw $p \sim \nu(\cdot)$\\
   \qquad Draw $L \sim \Psi(\cdot)$ and set $T = L\e$\\
   \qquad Propose $(x^*,p^*) \leftarrow \tilde{\varphi}_T(x^{i-1},p)$\\
   \qquad Set $x^i \leftarrow x^*$ with probability $1 \wedge \exp \left\{ H(x^{i-1},p) - H(x^*,p^*) \right\}$, \\ \qquad \quad otherwise set $x^i \leftarrow x^{i-1}$\\
   \enspace Output $x = (x^0,...,x^n)$
\end{tabbing}
\end{algorithm}

We work on the Borel space $(\R^{d},\B)$ and product space $\R^{2d}$ with the product $\sigma$-field. Throughout let $\|x\|,\|x\|_1$ and $\|x\|_\infty$ be the Euclidean, $L_1$ and $L_\infty$ norms of $x\in\R^{d}$, $B_{r}(x)=\{y\in\R^{d}:\|y-x\|< r\}$ be the Euclidean ball of radius $r$ centred at $x$, and $\partial_{i}=\partial/\partial x_{i}$. For functions $f,g:\R^d \to \R^d$, let $f \circ g(x) = f(g(x))$.  Whenever a distribution is referred to as having a density then this with respect to the relevant Lebesgue measure.  We write $x^i$ for the $i$th point in a Markov chain produced by Hamiltonian Monte Carlo, and $x_{i\e}$ for the $i$th step of the leapfrog integrator within a single iteration of the method.

\section{General observations}
\label{sec:general}

The effects of changing the proposal input noise distribution are more complex in Hamiltonian Monte Carlo than in the random walk Metropolis and the Metropolis-adjusted Langevin algorithm, two methods for which this question has been studied previously \citep{jarner2007convergence,stramer1999langevin}.  For both of these cases the input noise is combined linearly with some deterministic function of the current point of the Markov chain.  As a result, larger values of this noise will typically result in larger proposed moves, which may be advantageous if $\pi(\cdot)$ has heavy tails or possesses multiple modes.  By contrast, in Hamiltonian Monte Carlo the choice of kinetic energy alters both the input noise distribution and the Hamiltonian dynamics, and so the impact of different choices is not so transparent.

If we consider how a single proposal in Hamiltonian Monte Carlo perturbs the current position $x^{i-1}$ at iteration $i$, then setting $x_0 = x^{i-1}$, $p_0 \sim \nu(\cdot)$ and solving \eqref{eqn:hdynamics} for $T$ units of time gives
\begin{equation}
x^* = x_0 + \int_0^T \nabla K \left\{ p_0 - \int_0^s \nabla U(x_u)du \right\}ds.
\end{equation}
An explicit numerical method typically involves the approximation $\int_0^{h} \nabla U(x_u)du \approx \nabla U(x_0)h$ for some suitably chosen $h$.  In the case of the leapfrog integrator the approximate solution over a single leapfrog step becomes
\begin{equation} \label{eqn:compgrad}
x_\e = x_0 + \e \nabla K \left\{ p_0 - \frac{\e}{2} \nabla U(x_0) \right\},
\end{equation}
with the corresponding momentum update
\begin{equation} \label{eqn:momentum_update}
p_\e = p_0 - \frac{\e}{2} \left\{\nabla U(x_0) + \nabla U(x_\e) \right\}.
\end{equation}
Recalling that $K(p)$ is an even function, meaning $\nabla K(p)$ is odd, then it will turn out that in many cases the function that intuitively governs the speed of the $x$-dynamics when $\|\nabla U(x)\|$ is large is
\begin{equation} \label{eqn:compgrad2}
\nabla K \circ \nabla U(x),
\end{equation}
which we refer to from this point as the composite gradient.  Although in practice the choice of $\e$ will influence sampler behaviour, we will demonstrate that a qualitative understanding can be developed by considering only \eqref{eqn:compgrad2}.  Similarly, when $\|\nabla U(x)\|$ is small, $x_\e$ will resemble a random walk proposal with perturbation $\e$ multiplied by $\nabla K(p_0)$, which we refer to as the implicit noise.  We argue that an appropriate choice of kinetic energy is one for which these two quantities are both suitably optimized over, and in the next sections we examine both when this can be done and what a suitable choice for each should be.

\section{The composite gradient}
\subsection{Geometric convergence and long term stability}
We begin with some general results relating the choice of kinetic energy to the composite gradient, and potential consequences for the convergence rate of the method.

\begin{definition}
We call a distribution which has a density $f(x) \propto \exp\{-g(x)\}$ \emph{heavy-tailed} if $
\lim_{\|x\|\to\infty}\|\nabla g(x)\| = 0$, and \emph{light-tailed} if $\lim_{\|x\|\to\infty}\|\nabla g(x)\| = \infty$.
\end{definition}
Distributions which do not fall into either category are those for which the tails are heavy in some directions and light in others, for which the sampler will behave very differently in distinct regions of the space.  A detailed study of these is beyond the scope of this work, though we conjecture that the least favourable region would dictate convergence.  We also do not include distributions with exponentially decaying densities, which are often treated as a special case when analysing Metropolis--Hastings  methods \citep{jarner2000geometric}.

Recall that a $\pi$-invariant Markov chain with transition kernel $P$ is called geometrically ergodic if for some positive constants $C<\infty$, $\rho<1$ and a Lyapunov function $V:\R^d \to [1,\infty)$ the bound 
\begin{equation} \label{eqn:geoerg}
\|P^n(x,\cdot) - \pi(\cdot)\|_{TV} \leq CV(x)\rho^n
\end{equation}
can be constructed, where $\|\mu(\cdot) - \nu(\cdot)\|_{TV}$ denotes the total variation distance between two measures $\mu(\cdot)$ and $\nu(\cdot)$.  See for example \cite{roberts2004general} for more details.

\begin{proposition} \label{prop:lackge}
In either of the following cases Algorithm \ref{al1} will not produce a geometrically ergodic Markov chain:

(i) $\pi(\cdot)$ is heavy-tailed

(ii) $\pi(\cdot)$ is light-tailed and the following conditions hold:

\begin{enumerate}
\item The composite gradient satisfies
\begin{equation} \label{eq:cond1}
\lim_{\|x\|\to\infty}\frac{\|\nabla K\circ\nabla U(x)\|}{\|x\|} = \infty,
\end{equation}
\item There is a strictly increasing unbounded function $\phi:[0,\infty)\to [1,\infty)$
and $m_{1},m_{2}<\infty$ with $m_1 \geq 1$ such that for every $M\geq m_{1}$ and
$\|x\|\geq m_{2}$ 
\begin{equation}
\|y\|\geq M\|x\|\implies\|\nabla U(y)\|\geq\phi(M)\|\nabla U(x)\|.\label{eq:cond2}
\end{equation}
\item Setting $\triangle H(x_0,p_0)=H(x_{L\varepsilon},p_{L\varepsilon})-H(x_{0},p_{0})$
and $\triangle(x_0,p_0)=\|x_{L\varepsilon}\|-\|x_{0}\|+\|p_{L\varepsilon}\|-\|p_{0}\|$,
\begin{equation}
\triangle(x_0,p_0)\to\infty \implies \triangle H(x_0,p_0)\to\infty.\label{eq:cond3}
\end{equation}
\item For all $\delta < \infty$, there is a $C_\delta>0$ such that for all $x \in \R^d$
\begin{equation}
\frac{\|\nabla K \{\delta \nabla U(x)\}\|}{\|\nabla K \circ \nabla U(x)\|} \geq C_\delta
\label{eq:cond4}
\end{equation}
\item Either $K(p)$ satisfies
\begin{equation}
\|p\| \geq \|y\| \implies \|\nabla K(p)\| \geq \|\nabla K(y)\|,
\label{eq:cond5a}
\end{equation}
or there is a $k:\mathbb{R} \to \mathbb{R}$ such that $K(p) := \sum_{i=1}^d k(p_i)$ and
\begin{equation}
|p_i|\geq |y_i| \implies |k'(p_i)|\geq |k'(y_i)|.
\label{eq:cond5b}
\end{equation}
for all $i$.
\end{enumerate}

\end{proposition}

\begin{remark} Condition 4 in part (ii) will be satisfied if, for example, $\nabla K$ is a homogeneous function. It is satisfied for all choices considered in this work.
\end{remark}

For the standard quadratic choice of kinetic energy the Hamiltonian Monte Carlo method will not produce a geometrically ergodic Markov chain when $\pi(\cdot)$ has heavy-tails \citep{livingstone2016geometric}.  Part (i) of the above states that no choice of kinetic energy which is independent of $x$ can rectify this.  We explain why in subsection \ref{subsec:ht}.  Part (ii) states that provided the density is sufficiently regular that \eqref{eq:cond3} holds, and that the degree of tail oscillation in $\|\nabla U(x)\|$ is restricted through \eqref{eq:cond2}, then the composite gradient should not grow faster than linearly in $\|x\|$ for geometric ergodicity. Within the class of kinetic energies for which this requirement is met, we motivate an optimal choice in subsection \ref{subsec:epfamily}.

The composite gradient was deduced based on studying the leapfrog integrator over a single step. Next, however, we argue that controlling this quantity in fact induces stability over multiple leapfrog steps.
\begin{proposition} \label{prop:stability}
If for some $q>0$ and $A,B,C,D<\infty$
\begin{equation}
\|\nabla U(x)\| \leq A\|x\|^q + B, \qquad \|\nabla K(p)\| \leq C\|p\|^\frac{1}{q} + D,
\end{equation}
then $\limsup_{\|x\|\to\infty}\|\nabla K\circ\nabla U(x)\|/\|x\|<\infty$, and furthermore if $\|p_{0}\|\leq E_0\|x_0\|^q + F_0$ for some $E_0,F_0 <\infty$, then writing $x_{L\e}$ and $p_{L\e}$ as functions of $x_0$ and $p_0$, it holds that
\[
\|x_{L\e}(x_0,p_0)\| \leq E_L \|x_0\| + F_L, \qquad \|p_{L\e}(x_0,p_0)\| \leq \mathbf{G}_L\|x_0\|^q + \mathbf{H}_L.
\]
for some $E_L,F_L,\mathbf{G}_L,\mathbf{H}_L<\infty$ which do not depend on $x_0$.
\end{proposition}
As a contrast, if the assumption on $\nabla K(p)$ was replaced with $\|\nabla K(p)\|\leq C\|p\|+D$, then the bound would read $\|x_{L\varepsilon}(x_0,p_0)\| = E_L\|x_{0}\|^{1 \vee q^{L}} + F_L$, which grows exponentially in $L$ when $q>1$.

\subsection{Kinetic energies that induce negligible moves}
\label{subsec:negmoves}
Next we consider a heavy-tailed choice of $\nu(\cdot)$.  In this instance the composite gradient will in fact decay as $\|\nabla U(x)\|$ grows, meaning that if $\pi(\cdot)$ is light-tailed the resulting Markov chain will move very slowly in some regions of the space.  We formalize this intuition below.

\begin{definition}
We say that a Markov chain on $\R^d$ with transition kernel $P$ possesses the \emph{negligible moves} property if for every $\delta>0$
\begin{equation} \label{eqn:negmove}
\lim_{\|x\|\to\infty} P\{x,B_{\delta}(x)\} = 1.
\end{equation}
\end{definition}

\begin{proposition} \label{prop:negmove}
If $\pi(\cdot)$ is light-tailed and $\nu(\cdot)$ is heavy-tailed then a Markov chain produced by the Hamiltonian Monte Carlo method possesses the negligible moves property.
\end{proposition}

In general this property seems undesirable, the implication being that if the current point $x$ has large norm then the chain is unlikely to move very far in any appreciable number of steps.  It does not, however, preclude geometric ergodicity, as often such steps can be close to deterministic in a desirable direction.  We give a simple example to illustrate the phenomenon.

\begin{example}
Consider the Markov chain with transition
\[
x^i = \max(x^{i-1} - 1 + \xi^i,1), ~~\xi^i \sim N(0,1),
\]
which is geometrically ergodic as shown in section 16.1.3 of \cite{meyn1993markov}.  Then the transformed Markov chain $(y^i)_{i \geq 1}$, where each $y^i = \log(x^i)$, is also geometrically ergodic, but possesses the negligible moves property.
\end{example}

Although a geometrically decaying convergence bound can still be established in such cases, the slow movement of the chain when $\|x\|$ is large is accounted for in the Lyapunov function $V(x)$ present in the bound \eqref{eqn:geoerg}.

\begin{proposition} \label{prop:lyapunov}
If \eqref{eqn:negmove} holds for a Hamiltonian Monte Carlo method then any Lyapunov function $V:\X \to [1,\infty)$ cannot satisfy either of the following:

(i) $\log V$ is uniformly lower semi-continuous, meaning for all $\epsilon'>0$, there is a $\delta'<\infty$ such that if $\|y-x\|\leq\delta'$, uniformly in $x$, then $\log V(y)\geq\log V(x)-\epsilon'$.

(ii) There is an $s\in(0,\infty)$ such that $\liminf_{\|x\|\to\infty}V(x)e^{-s\|x\|}=c$ for some $c\in(0,\infty)$. 
\end{proposition}
The above implies that Lyapunov functions will typically exhibit faster than exponential growth in $\|x\|$, meaning the penalty in the bound \eqref{eqn:geoerg} for starting the chain in the tails of the distribution will be very large. This reflects the fact that the chain will take a long time to move an appreciable distance and hence approach equilibrium.  So although the negligible moves property does not preclude geometric convergence, we regard it as an undesirable feature.



\subsection{Illustrative example}
\label{subsec:epfamily}
Part (ii) of Proposition \ref{prop:lackge} relates to the numerical solution of \eqref{eqn:hdynamics} using the leapfrog method.  If the composite gradient grows faster than linearly in $\|x\|$ then the dynamical system exhibits stiffness, meaning gradients change very quickly and hence numerical approximations in which they are assumed to be constant over small time periods will become unstable.  Two possible remedies are to use a more complex implicit numerical method or adaptive step-size control, as discussed in \cite{okudo2016hamiltonian}.  With the standard choice of kinetic energy $\nabla K \circ \nabla U(x) = \nabla U(x)$, meaning these instabilities occur when $\|\nabla U(x)\|$ grows faster than linearly in $\|x\|$.  Since the composite gradient dictates this phenomenon in general, fast growth in $\nabla U(x)$ can be counteracted by choosing a kinetic energy which is subquadratic in $\|p\|$.  Physically, choosing a slowly growing kinetic energy  slows down the Hamiltonian dynamics, meaning the system \eqref{eqn:hdynamics} is no longer stiff in the $x$-coordinate and hence simple numerical schemes like the leapfrog method become stable.  We can further use the same idea when $\nabla U(x)$ exhibits relatively slow growth by choosing a fast-growing kinetic energy which will speed up the flow allowing the chain to explore the state space more efficiently.

To demonstrate this phenomenon more concretely we consider the family
of separable Hamiltonians $H:\R^{2}\to[0,\infty)$ given by
\begin{equation}
H(x,p)=\alpha^{-1}|x|^{\alpha}+\beta^{-1}|p|^{\beta},\label{eq:1d_expfamily}
\end{equation}
for suitable choices $\alpha,\beta>1$, which is necessary
to ensure continuity of the derivative. Hamilton's equations are
\begin{align}
\dot{x} & =\text{sgn}(p)|p|^{\beta-1}\nonumber \\
\dot{p} & =-\text{sgn}(x)|x|^{\alpha-1}\label{eq:1d_odes}
\end{align}
The flow is periodic since the contours of $H$ are closed curves
in $\R^{2}$. These kinetic energies correspond to choosing $\nu(\cdot)$ from the exponential power family of \cite{box1973bayesian}, for which direct sampling schemes are provided in \cite{mineo2014normalp}.  \cite{zhang2016towards} recommend choosing this class of momentum distributions in Hamiltonian Monte Carlo.  The following proposition shows how the the period length depends on the initial value of the Hamiltonian.

\begin{proposition} \label{prop:period}
For the class of Hamiltonians given
by \eqref{eq:1d_expfamily} and dynamics given by the exact solutions
to \eqref{eq:1d_odes}, the period length
\[
\mathcal{P}(E) = \inf_{t>0}\left\{t: (x_t,p_t)=(x_0,p_0)\right\}
\]
is related to the initial value of the Hamiltonian $E=H(x_0,p_0)$
through the expression 
\[
\mathcal{P}(E) = c_{\alpha,\beta}E^{\eta},~~~\eta=\frac{1-(\beta-1)(\alpha-1)}{\alpha\beta},
\]
where $c_{\alpha,\beta}<\infty$ does not depend
on $E$. In particular, if $\alpha = 1+\gamma$ and $\beta = 1+\gamma^{-1}$
for some $\gamma\in(0,\infty$), then $\nabla K\circ \nabla U(x) = x$ and $\eta=0$, meaning the period
length does not depend on $E$.
\end{proposition}

The above result characterizes precisely how the speed of the Hamiltonian flow depends on the growth rates of the potential and kinetic energies.  If $\eta <0$ the flow will be very fast at higher energies, and so explicit numerical methods will typically become unstable when $|x|$ is large.  If $\eta > 0$ then the period length will increase with $|x|$, so numerical methods will be more stable, but the sampler may take some time to move from the tails back into the centre of the space.

The ideal choice in terms of finding the right balance between efficiency and numerical stability would therefore appear to be setting $\eta = 0$, which for this class of models can be achieved by choosing the kinetic energy as the $L_p$ dual of the potential.  This corresponds to setting $\alpha = 1 + 1/(\beta - 1)$, which produces an algorithm for which the flow behaves in an auto-regressive manner when $|x|$ is large, since the composite gradient will be similar in magnitude to the current point.  In this case tuning the method should also be more straightforward, as a choice of integrator step-size and integration time made in the transient phase should still be reasonable at equilibrium.  When $\eta<0$, for any fixed choice of step-size $\e>0$ there is an $M_{\e} <\infty$ such that when $|x|>M_{\e}$ the algorithm will become numerically unstable, and when $\eta>0$ then it would be desirable in the transient phase to choose a comparatively large value for $\e$ and take more leapfrog steps, which may result in many proposals being rejected and unnecessary computational expense when the chain reaches equilibrium.





\section{The implicit noise}
\label{sec:implicitnoise}

There are likely to be regions in the centre of the space in which $\|\nabla U(x_0)\|$ is small, meaning the increment $\e \nabla K\{p_0 - \e\nabla U(x_0)/2\} \approx \e\nabla K(p_0)$.  This can occur when the density $\pi(x)$ contains flat regions, and we show one such case in section \ref{sec:examples}.  In these regions Hamiltonian Monte Carlo resembles a random walk Metropolis with proposal $x_{L\e} \approx x_0 + L\e\nabla K(p_0)$.  It would seem sensible to choose $K(p)$ such that the distribution of $\nabla K(p)$ would be an appropriate choice in a random walk Metropolis here.

For the standard quadratic choice $\nabla K(p) = p$, which is hence also Gaussian.  Different choices of kinetic energy, which may seem sensible by analysing the composite gradient, can however result in very different distributions for the implicit noise, which can either become bi-modal or exhibit very light tails.  We plot histograms of $\nabla K(p)$ for some natural choices of momentum distribution in figure \ref{fig1}.  

The optimum choice of implicit noise will typically be problem dependent, but it would seem logical to follow the standard choices used in the random walk Metropolis.  Bi-modal distributions have recently been suggested for this method \citep{yang2013searching}, but we found that something which resembles a Gaussian distribution for small values of $p$ performed favourably in experiments.

\begin{figure}
\includegraphics[height=10pc,width=30pc]{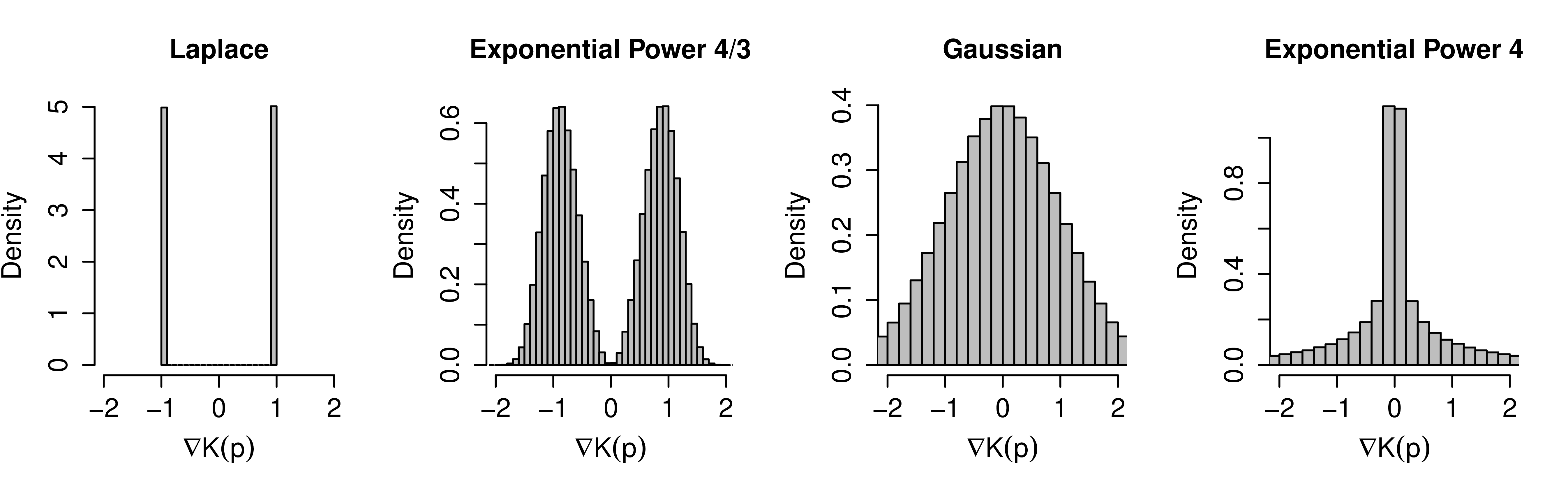}
\caption{Histograms showing how different choices of the momentum distribution $\nu(\cdot)$ can affect typical values of $\nabla K(p)$.}
\label{fig1}
\end{figure}

\section{General guidelines and the efficiency-robustness trade off}
\label{sec:guidelines}

The above results allow us to offer some general strategies and observations for choosing a kinetic energy, dependent on the objectives of the user.

If we deem the goal to be to maximize the speed of the flow when $\|x\|$ is large whilst retaining numerical stability, then an appropriate choice of kinetic energy for a given potential $U$ is one for which when $\|\nabla U(x)\|$ is large $\nabla K\{p_0 - \e \nabla U(x_0)/2\} \approx -c_{\e} x$ for some $c_{\e}>0$. This can be achieved when $\nabla U(x)$ is invertible through the choice
\begin{equation}
\nabla K (p) \approx \nabla U^{-1} \left(p\right),
\end{equation}
which implies that $\nabla K \circ \nabla U(x) \approx \nabla U^{-1} \circ \nabla U (x) = x$.  Of course often $\nabla U(x)$ will not possess an inverse, in which case the considering the leading order term in $\|\nabla U(x)\|$ is of greatest importance.  Below we give some examples of how such a choice for $K(p)$ can be constructed.

\begin{example}
If $x \sim N(0,\Sigma)$ then $\nabla U(x) = \Sigma^{-1}x$, meaning $-(\nabla U)^{-1}(-p) = \Sigma p$, implying the choice $p \sim N(0,\Sigma^{-1})$ and corresponding kinetic energy $K(p) = p^T\Sigma p/2$.
\end{example}

\begin{remark}
This result is intuitive, as if $K(p) = p^TM^{-1}p/2$, then the proposal after a single leapfrog step will be $x_\e = x_0 - \e^2 M^{-1} \nabla U(x) / 2 + \e M^{-1/2}z$ with $z \sim N(0,I)$.  This is equivalent to a Metropolis-adjusted Langevin algorithm proposal with pre-conditioning matrix $M^{-1}$, and in that setting it is known that this should be made equal to $\Sigma$ for optimal sampling \citep{roberts2001optimal,dalalyan2017theoretical}.
\end{remark}

\begin{example}
Consider a distribution with density $\pi(x) \propto e^{-U(x)}$ where
\[
U(x) = C(\alpha x^4 - \beta x^2),
\]
for some $C,\alpha,\beta>0$.  Such double well potential models are important in many areas of physics, including quantum  mechanical tunnelling~\citep{liang1992periodic}.  In this case $\nabla U(x) = 4C\alpha x^3 - 2C\beta x$, which is typically not invertible, but for large $|x|$ we see that $\nabla U(x)/x^3 \approx 4C\alpha$, which motivates the choice $\nabla K(p) = p^{1/3}/(4C\alpha)$ for large $|p|$.
\end{example}

\begin{example} \cite{park2008bayesian} suggest a Bayesian version of the bridge estimator for regression of \cite{frank1993statistical}.  In this case
\[
U(\beta) = L(\beta) + \lambda \sum_{j=1}^d |\beta_j|^q,
\]
for some $\lambda>0$, where $L(\beta)$ is a suitable loss function.  If $q >1$ and the asymptotic growth of $L(\beta)$ is at most linear in $\|\beta\|$, then the choice $K(p) = \sum_i|p_i|^\alpha$ with $\alpha = 1+1/(q-1)$ for large $\|p\|$ would maximize the speed of convergence during the transient phase of the algorithm.  Logistic regression, the Huberized loss of \cite{park2008bayesian} and the asymmetric pinball loss \citep{yu2001bayesian} are possible choices for $L(\beta)$ which fit this criterion.
\end{example}

\begin{example} \cite{neal2003slice} proposes the hierarchical model $x_i|\nu \sim N(0,e^{\nu})$ for $i = 1,...,9$, with $\nu \sim N(0,3^2)$, resulting in a joint distribution that resembles a ten-dimensional funnel.  Here $U(x_1,...,x_9,\nu) = \nu^2/18 + 9\nu/2 + e^{-\nu}\sum_i x_i^2 / 2$, meaning
\[
\nabla U(x_1,...,x_9,\nu) = \left( 2x_1 e^{-\nu}, ..., 2x_9 e^{-\nu}, \frac{\nu}{2} + \frac{9}{2} - \frac{e^{-\nu}}{2}\sum_i x_i^2 \right)^T.
\]
When $\nu \ll 0$ these gradients can be exponentially large, so a choice of $\nabla K(p)$ that is at most logarithmic in $p$ would be needed, which could arise from the choice of kinetic energy $K(p_i) \approx \{1 + \log(1 + |p_i|)\}|p_i|$.
\end{example}


A general appreciation for the consequences of different kinetic energy choices can be gained by considering a formal Taylor series expansion
\[
\nabla K\left\{ p_0 - \frac{\e}{2}\nabla U(x_0)\right\} = \nabla K\left\{ - \frac{\e}{2}\nabla U(x_0)\right\} + \nabla^2 K\left\{ - \frac{\e}{2}\nabla U(x_0)\right\}p_0 + \cdots
\]
In the Gaussian case $\nabla^2 K(p) = I$, meaning the refreshed momentum $p_0$ interacts linearly with the composite gradient term regardless of the current position in the state space.  If however, the momentum distribution is chosen to have heavier than Gaussian tails, then each element of $\nabla^2 K(p)$ will become negligibly small when $\|\nabla U(x_0)\|$ is large, reducing the influence of $p_0$.  A sampler with a flow of this type will move towards a high-density region of the space in an almost deterministic fashion when $\|x\|$ is large.  We formalize this intuition below.

\begin{definition}
A Hamiltonian Monte Carlo sampler is called \emph{deterministic in the tails} if as $\|x\| \to \infty$
\begin{equation}
\left\|\nabla^2 K\left\{ \frac{\e}{2}\nabla U(x) - p\right\} - \nabla^2 K \left\{ \frac{\e}{2} \nabla U(x) \right\} \right\| \to 0
\end{equation}
a.s., where $p \sim \nu(\cdot)$.

\end{definition}

Given the above intuition, the following result is immediate.

\begin{proposition} \label{prop:robust}
A Hamiltonian Monte Carlo sampler will be deterministic in the tails if $\pi(\cdot)$ is light-tailed and $K(p)$ is twice differentiable with
\begin{equation} \label{eqn:robust}
\lim_{\|p\| \to\infty} \| \nabla^2 K(p)\| = 0.
\end{equation}
\end{proposition}

\begin{remark}
Typically if a choice of momentum distribution with heavier than Gaussian tails is made then \eqref{eqn:robust} will hold.
\end{remark}

As alluded to in subsection \ref{subsec:epfamily}, there is a trade off to be made between efficiency and step-size robustness. Choosing a lighter-tailed momentum distribution will result in a faster flow, and potentially a faster mixing Markov chain as a result, but the risk of numerical instabilities is higher, and so the algorithm will be very sensitive to the choice of $\e$.  Choosing a heavier-tailed momentum distribution will result in slower flows, meaning algorithm performance will be less sensitive to step-size choice, but the best-case mixing time under the optimal step-size may be slower.  We give a numerical example illustrating this in the supplementary material.  If the tails are too heavy then the algorithm will also exhibit the negligible moves property of section \ref{subsec:negmoves}.  A kinetic energy choice with Laplacian tails, for example, would result in a sampler which may be comparatively slow in the tails but is guaranteed to be numerically stable, since in this case the composite gradient is $\text{sign}\{ p_0 - \e\nabla U(x_0)/2\}$, which is bounded from above and below regardless of how fast $\|\nabla U(x_0)\|$ grows.

\begin{remark}
It has recently been shown that by choosing $\nu(\cdot)$ to be a Laplace distribution one can perform Hamiltonian Monte Carlo on problems for which the standard method is not applicable, such as problems with discontinuous likelihoods or involving discrete parameters. We do not comment further on this here, but see \cite{nishimura2017discontinuous} for details.
\end{remark}

\subsection{Relativistic power kinetic energies}
\label{subsec:grfamily}
From the above discussion it would seem sensible to design a kinetic energy function that looks similar to the Gaussian choice when its argument is small, but also allows a robust and efficient choice of composite gradient.  The relativistic kinetic energy 
\[
K(p) = \sum_i m_ic_i^2\left( 1+\frac{p_i^2}{m_i^2 c_i^2} \right)^{\frac{1}{2}}
\]
suggested in \cite{lu2016relativistic}, has many such desirable features, since it behaves as a quadratic function for small $p$ and is linear when $p$ is large.  We slightly generalize this choice to allow for different tail behaviours, resulting in the family of relativistic power kinetic energies
\begin{equation}
K(p) = \sum_i \beta^{-1} \left( 1+ \gamma_i^{-1}p_i^2 \right)^{\beta/2},
\end{equation}
where $\beta \geq 1$ and each $\gamma_i >0$.  When each $p_i^2 \ll \gamma_i$ then 
\[
K(p) \approx \sum_i \beta^{-1}(1+\gamma_i^{-1}p_i^2) \implies \partial_i K(p) \approx 2 (\beta\gamma_i)^{-1}p_i
\]
and when each $p_i^2 \gg \gamma_i$
\[
K(p) \approx \sum_i \beta^{-1} \gamma_i^{-\beta/2}|p_i|^{\beta} \implies \partial_i K(p) \approx \gamma_i^{-\beta/2}\text{sgn}(p)|p_i|^{\beta - 1}.
\]

Drawing samples from this family is a requirement for implementation. In our experiments we used adaptive rejection sampling ~\citep{gilks1992adaptive}, which we found to be reasonably efficient.  However, since the distribution is fairly regular in shape it is likely that specialized methods can be designed to be even more efficient. In addition, other strategies such a dependent sampling using Metropolis--Hastings could be employed, as in \cite{akhmatskaya2008gshmc}.

\subsection{Approximate gradients and doing away with Metropolis--Hastings}

One example in which the robustness-efficiency trade off outlined above can be informative is when estimates are used in place of $\nabla U(x)$, which may be either intractable or very expensive to compute.  Examples are given in \cite{chen2014stochastic,strathmann2015gradient,lindsten2016pseudo}.

If we assume that at each iteration the approximate gradient $\tilde{\nabla}U(x)$ satisfies
\[
\tilde{\nabla} U(x) = \nabla U(x) + \eta_x,
\]
where $\eta_x$ follows a distribution that may depend on $x$, but such that $E(\eta_x)=0$ for all $x$, then a similar Taylor expansion gives
\[
\nabla K\left\{ p_0 - \frac{\e}{2}\tilde{\nabla} U(x_0)\right\} = \nabla K\left\{ - \frac{\e}{2}\nabla U(x_0)\right\} + \nabla^2 K\left\{ - \frac{\e}{2}\nabla U(x_0)\right\}(p_0+\eta_x) + \cdots.
\]
Therefore, if the momentum distribution is chosen to have heavy enough tails that for large $\|x\|$ the last term becomes negligibly small, then the effects of such an approximation are mitigated and the the resulting approximate composite gradient will closely resemble $\nabla K\left\{ - \frac{\e}{2}\nabla U(x_0)\right\}$.

In some of the above mentioned approximate implementations, the Metropolis step is also omitted to reduce computational costs. When the gradient term $\|\nabla U(x)\|$ grows at a faster than linear rate, then often the resulting Markov chains become transient when this is done, as shown for the unadjusted Langevin algorithm in \cite{roberts1996exponential}.  In such cases it is difficult to give long time guarantees on the level of bias induced from a finite simulation.  Ensuring that the composite gradient grows no faster than linearly should mean that such transience is averted, meaning a stronger grasp of the degree of approximation can be established.  We leave a detailed exploration of this for future work.

\subsection{Heavy-tailed models}
\label{subsec:ht}

Based on the composite gradient intuition, it may seem desirable when $\pi(\cdot)$ is heavy-tailed to choose a kinetic energy for which $\|\nabla K(p)\| \to \infty$ as $\|p\| \to 0$.  In this case when $\|x\|$ is large the potential gradient will be very small, so choosing the kinetic energy in this way one could still make $\nabla K \circ \nabla U(x)$ linear in $\|x\|$.  However, the actual leapfrog perturbation is
\[
\e \nabla K \left\{ p_0 - \frac{\e}{2}\nabla U(x_0) \right\}.
\]
When $\pi(\cdot)$ is light-tailed the $p_0$ term becomes safe to ignore when $\|x\|$ is large, but if it is heavy-tailed then this is no longer the case as for large $\|x\|$ typical proposals will be of the form $x^* \approx x + L\e\nabla K( p )$, where $p \sim \nu(\cdot)$, which resembles a random walk.  It is because of this that, as shown in Part (i) of Proposition \ref{prop:lackge}, no choice of kinetic energy which is independent of $x$ can produce a geometrically ergodic chain.

It will likely be the case that choosing a heavier-tailed momentum distribution will be of benefit in this scenario in terms of rates of convergence, since this is true of the random walk Metropolis, as shown in \cite{jarner2007convergence}.  However, since it is the implicit noise $\nabla K(p)$ which drives the sampler, then care must also be taken with its form.  In the case where $\nu(\cdot)$ is a Cauchy distribution, for example, then $\nabla K(p) = 2p/(1+p^2)$, which is lighter-tailed than $p$ itself, so here the benefits of choosing heavier tails are not as strong as those for the random walk algorithm.

\section{Examples}
\label{sec:examples}

\subsection{Quantile regression}

We consider a Bayesian quantile regression as introduced in \cite{yu2001bayesian}.  The goal is to estimate the $\tau$th quantile
of a response $y\in\R$ conditioned on a collection of covariates
$x\in\R^m$, written $\mu(x)=F^{-1}(\tau|x)$.  Given $n$ data points $\{(x_{i},y_{i})\}_{i=1}^{n}$, we take the natural choice $\mu(x_i,\beta) = \sum_{j=1}^m x_{ij}\beta_j$,  and follow the approach of \cite{fasiolo2017fast} by estimating a posterior distribution for $\beta$ as $\pi(\beta|y,x)\propto\exp\left\{ -\sum_{i=1}^{n}L(\beta,x_{i},y_{i})\right\} \pi_{0}(\beta)$, using the general Bayesian updating framework of \cite{bissiri2016general}, with $L(\beta,x_i,y_i)$ a smoothed version of the pinball loss introduced in \cite{fasiolo2017fast}, given by
\[
L(\beta,x,y)=(\tau-1)\left\{ \frac{y-\mu(x,\beta)}{\sigma}\right\} +\xi\log\left[1+\exp\left\{ \frac{y-\mu(x,\beta)}{\xi\sigma}\right\} \right]+g(\xi,\sigma,\tau).
\]
Here $g(\xi,\sigma,\tau) = \log\left\{ \xi\sigma\text{Beta}\left[\xi(1-\tau),\xi\tau\right]\right\}$.  As $\xi \to 0$ then the non-smooth pinball loss is recovered.  This is linear in $\beta$, meaning that if $L_{q}$ priors are chosen as in bridge regression \citep{park2008bayesian} then the corresponding potential is
\[
U(\beta)=\sum_{i=1}^{n}L(\beta,x_{i},y_{i})+\lambda\sum_{j=1}^{d}|\beta_{j}|^{q},
\]
for $1<q\leq2$. The dominant term in $U(\beta)$ is the contribution from the prior, meaning that each $|\partial_j U(\beta)| = O(|\beta_j|^{q-1})$ for $j = 1,...,m$.

We performed two studies on 20 simulated data points with $m=2$, fixing $\sigma = 1$, $\lambda = 1$ and $\xi = 0.01$.  In the first we set $q=2$ for each $\beta_j$, and in the second $q = 1.5$.  Four different momentum distributions were tested: an exponential power family with shape parameter 3, Gaussian, Laplacian, and a t distribution with 4 degrees of freedom.  Although the set up is simple it still enables a demonstration of both the negligible moves property and how the composite gradient dictates sampler performance.  The number of leapfrog steps was set to 1 and the step size was tuned based on achieving a 65-75\% acceptance rate at equilibrium.  The samplers were then initialised far from the region of high probability and convergence speed was assessed, and also initialised at equilibrium to assess mixing in favourable areas.  Results are shown in Figure \ref{fig2}.  In the Gaussian prior study the exponential power choice is not shown since the resulting sampler did not move, as explained by Proposition \ref{prop:lackge}.  As expected, the Gaussian choice, which results in a linearly growing composite gradient, reaches equilibrium quicker than the others.  The Laplacian choice converges in a straight-line fashion, as the size of the proposed jump is always $\e\surd 2$ regardless of the current position.  The Student's t choice exhibits the negligible moves property as outlined in Proposition \ref{prop:negmove}, with convergence visibly slow during the first 5000 iterations.  When the heavier-tailed prior is chosen, the exponential power choice is now numerically stable and produces faster convergence than the Gaussian.  Since the potential growth is $O(|\beta|^{1.5})$ then this option results in a linear composite gradient, so the performance is to be expected.  The speed of the remaining choices is dictated by how the composite gradient grows, with the Student's t distribution again performing the worst.  In this study the difference is not so pronounced, as in probabilistic terms the sampler is initialised in a region that is not so far into the tails.  The slower convergence overall compared to the Gaussian study reflects the fact that $\pi(\cdot)$ is heavier-tailed.  In both studies all samplers mixed similarly well regardless of the kinetic energy when initialised at equilibrium, with slightly worse performance observed for the Student's t distribution.  As discussed in Section \ref{sec:implicitnoise}, the optimal choice of implicit noise is likely to be problem specific, with the Gaussian shape sensible when no other information is available.

\begin{figure}
\includegraphics[height=10pc]{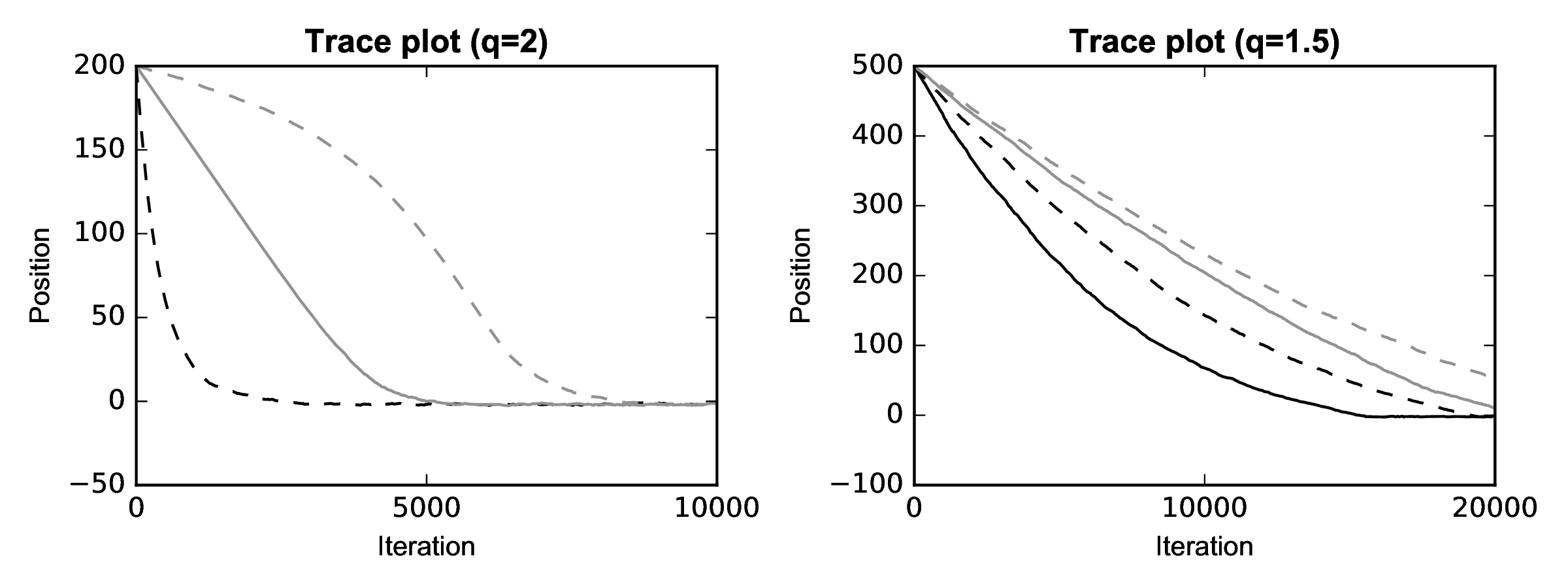}
\caption{Trace plots from the quantile regression studies. The solid black line is the exponential power momentum choice, dashed black lines are the Gaussian choice, solid grey lines are the Laplacian choice and dashed grey lines are the t distribution with 4 degrees of freedom.}
\label{fig2}
\end{figure}

\subsection{The Ginzburg--Landau model}

As a second example we take the model of phase transitions in condensed-matter physics proposed by Ginzburg and Landau \citep[Chapter~5]{goldenfeld1992lectures}.  We consider a three-dimensional $10^3$ lattice, where each site represents a random variable $\psi_{ijk} \in \R$.  The corresponding potential for the 1000-dimensional parameter $\psi$ is
\[
U(\psi) = \sum_{i,j,k} \left\{ \frac{(1-\tau)}{2} \psi_{ijk}^2 + \frac{\tau\alpha}{2} \|\tilde{\nabla}\psi_{ijk}\|^2 + \frac{\tau\lambda}{4} \psi_{ijk}^4 \right\},
\]
where $\alpha,\lambda, \tau > 0$ and $\tilde{\nabla}\psi_{ijk} = (\psi_{i_+ jk}-\psi_{ijk},\psi_{ij_+k}-\psi_{ijk},\psi_{ijk_+}-\psi_{ijk})$, where $i_+=(i+1)\mod(10)$, and $j_+,k_+$ are defined analogously.  The modular structure imposes periodic boundary conditions on the lattice.  When the parameter $\tau < 1$, the conditional distribution for each $\psi_{ijk}$ looks Gaussian in the centre of the space but with lighter tails.  The model exhibits a phase transition at $\tau = 1$, so that when $\tau > 1$ each conditional distribution is bi-modal.  The parameter $\alpha$ controls the strength of correlation between neighbouring lattice sites, with larger values making the sampling problem more challenging.  When $\tau < 1$, each $\psi_{ijk}$ is likely to be close to zero, while in the bi-modal phase they are more likely to be non-negligible in magnitude. When $\alpha$ is large, neighbouring parameters are likely to have the same sign.  The bi-modal phase therefore represents the system in its ordered state, whereas the system is disordered for $\tau <1$.

The inference problem is to estimate expectations with respect to the distribution with density proportional to $\exp\{ -U(\psi) \}$.  We generate samples using Hamiltonian Monte Carlo with four different choices of kinetic energy: the standard quadratic choice, the relativistic choice of \cite{lu2016relativistic}, the exponential power choice of \cite{zhang2016towards} and the relativistic power choice introduced in section \ref{subsec:grfamily}.  In the latter two cases we set the shape parameter $\beta = 4/3$, which results in a composite gradient that is linear when $\|\psi\|$ is large.

We perform two studies.  In the first we initialize samplers in the centre of the space and compute the effective sample sizes from 10,000 iterations.  In the second we initialize each $\psi_{ijk} \sim U[-10,10]$, and compute the number of iterations until $\max |\psi_{ijk}| \leq 2$.  Results are shown in Table \ref{table1}.  The kinetic energies with Gaussian implicit noise produce larger effective sample sizes in this example.  The reason is that there is a sharp drop in density when the quartic term begins to dominate the potential.  As the exponential power choice results in a bi-modal implicit noise, it is more likely to propose larger jumps which fall into this low density region, and are hence rejected.  To prevent this from happening too often a smaller step-size must be chosen than for the other methods.  Among the other choices the Gaussian has the highest effective sample size, though this cannot necessarily be relied upon as the sampler is not geometrically ergodic.  The slower speed for larger values of $\|\psi\|$ enforced in the relativistic and to a lesser degree relativistic power cases appears to slightly reduce efficiency here, which is sacrificed in favour of robustness.  We simply set each tuning parameter $\gamma_i = 1$ in our experiments.  This could be tuned in order to improve performance, but we preferred to use default values to limit the extra tuning required for this choice.  Sensitivity checks with $\gamma_i = 2$ and $3$ did not yield appreciably different results.  In terms of iterations to the centre, the Gaussian choice performs poorly as expected.  The remaining samplers perform similarly.  The relativistic power choice outperforms the exponential power choice because a smaller step-size was needed in the latter case for adequate mixing at equilibrium.  The relativistic choice, though slower, is not critically so, because $10$ is not very large in absolute terms, meaning the differences between a drift which is $O(1)$ and $O(\|\psi\|)$, while visible, is not substantial here.

\begin{table}
\centering
\def~{\hphantom{0}}
\begin{tabular}{lcccc}
 \\
& \multicolumn{3}{c}{Effective sample size} & Iterations to centre  \\[5pt]
& Minimum & Mean & Maximum & \\
Gaussian 							& 6,251 & 8,748 & 10,000  & N/A  \\
Relativistic Power $\beta = 4/3$ 	& 5,253 & 6,777 & 8,271   & 4.2  \\
Relativistic 						& 3,591 & 4,639 & 5,525   & 8.6 \\
Exponential Power $\beta = 4/3$ 	& 810 	& 1,108 & 1,303    & 11.9
\end{tabular}

\caption{ \label{table1}
Effective sample sizes at equilibrium and iterations until $\max_i|\psi_i|<2$ for Hamiltonian Monte Carlo on the Ginzburg--Landau model with $\alpha=0.1$, $\lambda=0.5$ and $\tau=2$.  Step sizes were chosen based on optimizing effective sample sizes.  In each simulation the number of leapfrog steps was set to 10, and results are averaged over 10 runs.}

\end{table}

\section{Discussion}
\label{sec:discussion}

We have described how changing the kinetic energy in Hamiltonian Monte Carlo affects performance.  In practice, several other strategies could be considered, such as mixtures of kinetic energies, or employing a delayed rejection approach as in \cite{mira2001metropolis}. In addition, other stochastic processes that utilise Hamiltonian dynamics, such as underdamped Langevin diffusions, can also be considered \citep{stoltz2016stable}.

In this work kinetic energies of the form $K(p) = \sum_i K_i(p_i)$ were used. The reason for this is that $\partial_i K(p)=\partial_i K_i(p_i)$, meaning that the when $\pi(x)$ is of the product form the Hamiltonian flows become independent. If a different choice were made, for example $K(p) = \beta^{-1}(1+\gamma^{-1}p^T p)^{\beta/2}$, then $\partial_i K(p) = (1+\gamma^{-1}p^{T}p)^{\beta/2-1}p_i$, so this does not always happen, here only when $\beta = 2$. It is unclear which of these approaches is preferable from a theoretical perspective, but evaluating $p^Tp$ is certainly a computational burden, and previous empirical evidence favours our approach \citep{lu2016relativistic}.

Outside the realm of separable Hamiltonians, in \cite{girolami2011riemann} the choice $K(x,p) = \{ \log|G(x)| + p^T G^{-1}(x) p \} / 2$ is advised.  Since an implicit integrator is required, the composite gradient intuition no longer directly applies.  Here instead
\[
x_{\e} = x_0 + \frac{\e}{2}\left[ \left\{ G^{-1}(x_0) + G^{-1}(x_{\e}) \right\} 
\left\{ p_0 - \frac{\e}{2}\nabla U(x_0) - \frac{\e}{2} \nabla f(x_0,p_{\frac{\e}{2}}) \right\} \right],
\]
where $G(x)$ is some Riemannian metric and $2f(x,p) = \log|G(x)| + p^T G^{-1}(x)p$.  Key drivers of the dynamics are the terms $G^{-1}(x_0)\nabla U(x_0)$ and $G^{-1}(x_{\e})\nabla U(x_0)$, which are often called the natural gradient \citep{amari2007methods}.  Given recent results \citep{taylor2015exact}, this approach may be advantageous when $\pi(\cdot)$ has heavy-tails.  The log-determinant term can also be beneficial \citep{betancourt2015hamiltonian}.

Another line of further study is to assess how kinetic energy choice affects dynamic implementations such as the No-U-Turn Sampler \citep{hoffman2014no}.  Results in simple cases suggest that such implementations are favourable when $\pi(\cdot)$ is heavy-tailed \citep{livingstone2016geometric}.

\section*{Acknowledgement}
We thank the reviewers for helping to improve the paper. SL thanks Michael Betancourt for sharing a calculation related to Proposition \ref{prop:period}, Matteo Fasiolo for advice on quantile regression, and Gabriel Stoltz, Zofia Trstanova and Paul Fearnhead for useful discussions. The authors acknowledge support from the Engineering and Physical Sciences Research Council for support through grants EP/K014463/1, i-like (SL and GOR) and EP/D002060/1, CRiSM (GOR).


\bibliographystyle{biometrika}
\bibliography{paper-ref}

\appendix

\section{Proofs of propositions}
\label{app:prop}

\begin{proof}[Proof of proposition \ref{prop:lackge}]
(i) In Lemma B\ref{lem:ht1} we show that $\int e^{\gamma\|x\|-U(x)}dx = \infty$ for any $\gamma>0$,
and in Lemma B\ref{lem:ht2} that for any $\eta>0$ there is an $r<\infty$
such that $P\{x,B_r(x)\}>1-\eta$. Theorem 2.2 of \cite{jarner2003necessary}
establishes that if these two conditions hold then the resulting Markov
chain cannot be geometrically ergodic.

(ii) Lemmas B\ref{lem:1stepbound}, B\ref{lem:Lstepbound} and B\ref{lem:pbound}
show that when \eqref{eq:cond1}-\eqref{eq:cond5b}  hold, with
probability one $\lim_{\|x_{0}\|\to\infty}\triangle(x_{0},p_{0})=\infty$,
where $\triangle(x_{0,}p_{0})=(\|x_{L\varepsilon}\|+\|p_{L\varepsilon}\|)-(\|x_{0}\|+\|p_{0}\|)$.
Under \eqref{eq:cond3} this implies that with probability one $\lim_{\|x_{0}\|\to\infty}\triangle H(x_{0},p_{0})=\infty$,
where $\triangle H(x_{0},p_{0})=H(x_{L\varepsilon},p_{L\varepsilon})-H(x_{0},p_{0})$.
This in turn implies that with probability one $\lim_{\|x_{0}\|\to\infty}\alpha(x_{0},x_{L\varepsilon})=0$,
which, using Proposition 5.1 of \cite{roberts1996geometric}, establishes the result.
\end{proof}
\begin{proof}[Proof of proposition \ref{prop:stability}]
For the first part, note that the assumptions imply 
\begin{align*}
\|\nabla K\circ\nabla U(x)\|\leq C(A\|x\|^{q}+B)^{1/q}+D,
\end{align*}
which implies $\limsup_{\|x\|\to\infty}\|\nabla K\circ\nabla U(x)\|/\|x\|<\infty$
as required. We prove the second part by induction. Precisely, we
show that assuming $\|p_{i\e}\|\leq E_{i}\|x_{i\e}\|^{q}+F_{i}$
for some $E_{i},F_{i}<\infty$ implies $\|p_{(i+1)\e}\|\leq E_{i+1}\|x_{i\e}\|^{q}+F_{i+1}$
and $\|x_{(i+1)\e}\|\leq G_{i}\|x_{i\e}\|+H_{i}$ for $E_{i+1},F_{i+1},G_{i},H_{i}<\infty$.
These in turn imply the result. First note that 
\begin{align*}
\|x_{(i+1)\e}-x_{i\e}\| & =\varepsilon\|\nabla K\{p_{i\e}-\frac{\varepsilon}{2}\nabla U(x_{i\e})\}\|\\
 & \leq\varepsilon C\|p_{i\e}-\frac{\varepsilon}{2}\nabla U(x_{i\e})\|^{1/q}+\varepsilon D\\
 & \leq\varepsilon C\left\{ \|p_{i\e}\|+\frac{\varepsilon}{2}\|\nabla U(x_{i\e})\|\right\} ^{1/q}+\varepsilon D.
\end{align*}
Using $\|\nabla U(x_{i\e})\|\leq A\|x_{i\e}\|^{q}+B$ gives 
\begin{align*}
\|x_{(i+1)\e}\| & \leq\|x_{i\e}\|+\varepsilon C\left\{ (E_{i}+\varepsilon A/2)\|x_{i\e}\|^{q}+\varepsilon B/2\right\} ^{1/q}+\varepsilon D.
\end{align*}
Given this we can choose $G_{i}=\varepsilon C(E_{i}+\varepsilon A/2+\varepsilon B/2)^{1/q}+1$
and $H_{i}=\varepsilon C(E_{i}+\varepsilon A/2+\varepsilon B/2)^{1/q}+\varepsilon D$
to see that
\[
\|x_{(i+1)\e}\|\leq G_{i}\|x_{i\e}\|+H_{i}.
\]
Iterating gives
\[
\|x_{(i+1)\e}\| \leq \mathbf{G}_L \|x_0\| + \mathbf{H}_L,
\]
where $\mathbf{G}_L = G_{L-1}G_{L-2}...G_0$ and $\mathbf{H}_L = H_{L-1} + G_{L-1}H_{L-2} + G_{L-1}G_{L-2}H_{L-3} + ... + G_{L-1}...G_1H_0$.  Next recall that
\begin{align*}
\|p_{(i+1)\e}-p_{i\e}\| & =\frac{\varepsilon}{2}\|\nabla U(x_{i\e})+\nabla U(x_{(i+1)\e})\|\\
 & \leq\frac{\varepsilon}{2}\left\{ \|\nabla U(x_{i\e})\|+\|\nabla U(x_{(i+1)\e})\|\right\} \\
 & \leq\frac{\varepsilon}{2}\left(A\|x_{i\e}\|^{q}+A\|x_{(i+1)\e}\|^{q}+2B\right)\\
 & \leq\frac{\varepsilon}{2}\left\{ A\|x_{i\e}\|^{q}+A(G_{i}\|x_{i\e}\|+H_{i})^{q}+2B\right\} \\
 & \leq\frac{\varepsilon}{2}\left[A\left\{ 1+(G_{i}+H_{i})^{q}\right\} \|x_{i\e}\|^{q}+A(G_{i}+H_{i})^{q}+2B\right].
\end{align*}
Combining with the assumption that $\|p_{i\e}\|\leq E_{i}\|x_{i\e}\|^{q}+F_{i}$,
gives
\[
\|p_{(i+1)\e}\|\leq\left[E_{i}+\frac{\varepsilon}{2}A\left\{ 1+(G_{i}+H_{i})^{q}\right\} \right]\|x_{i\e}\|^{q}+\frac{\varepsilon}{2}\left\{ A(G_{i}+H_{i})^{q}+2B\right\} +F_{i}.
\]
Setting $E_{i+1}=\left[E_{i}+\varepsilon A\left\{ 1+(G_{i}+H_{i})^{q}/2\right\} \right]$
and $F_{i+1}=\varepsilon\left\{ A(G_{i}+H_{i})^{q}+2B\right\} /2+F_{i}$
then gives $\|p_{(i+1)\e}\|\leq E_{i+1}\|x_{i\e}\|^q+F_{i+1}$. Iterating then gives $\|p_{L\e}\| \leq E_L\|x_0\|^q+F_L$. Recalling that $\|p_0\| \leq E_0\|x_0\|^q + F_0$ by assumption completes the proof.
\end{proof}

\begin{proof}[Proof of proposition \ref{prop:negmove}]
Consider the event $B = \{ 4\|p_0\| \leq \e\|\nabla U(x_0)\| \}$, and note that $\lim_{\|x\| \to \infty} \text{pr}(B) = 1$.  We use the facts that  $\|x_{L\e}-x_{0}\|\leq\sum_{i=1}^{L-1}\|x_{(i+1)\e}-x_{i\e}\|$, and that for any $i \in \{0,...,L-1\}$
\begin{equation} \label{eqn:diff}
\| x_{(i+1)\e} - x_{i\e} \| = \e \|\nabla K( p_{\frac{2i+1}{2}\e} ) \|.
\end{equation}
Taking $i=0$ gives
\[
\| x_\e - x_0 \| = \e \|\nabla K \left\{ p_0 - \e\nabla U(x_0)/2 \right\} \|.
\]
Since $4\| p_0 - \e\nabla U(x_0)/2\| \geq \e\|\nabla U(x_0)\|$ under $B$, it follows from the fact that $\pi(\cdot)$ is light-tailed and $\nu(\cdot)$ heavy-tailed that for every $\delta >0$ there is an $M<\infty$ such that whenever $\|x_0\| > M$ then $\| \nabla K(p_{\e/2}) \| < \delta/\e$.  Thus $\|x_\e - x_0\|$ can be made arbitrarily small by choosing an $x_0$ with large enough norm.

Recall that $\nabla U(x)$ is continuous by assumption.  It follows from the preceding argument that for any $\gamma_1>0$ we can choose an $x_0$ with large enough norm that $\|\nabla U(x_\e) - \nabla U(x_0)\| < \gamma_1$ under $B$.

To complete the proof we show that if $\sum_{j=1}^{i}\|x_{j\e} - x_{(j-1)\e}\| < \delta/2$ then $\|x_{(i+1)\e} - x_{i\e} \| \leq \delta/2$ under $B$.  Combining this with the previous paragraphs establishes that for any $\delta>0$ then there is an $x_0$ with large enough norm that $\|x_{L\e} - x_0\|<\delta$ if event $B$ holds, establishing the result.

From equation \eqref{eqn:diff} the key factor in controlling $\|x_{(i+1)\e} - x_{i\e}\|$ is $\|p_{(i+1/2)\e}\|$, which can be lower bounded using
\begin{equation}
\|p_{(i+1/2)\e}\| \geq \frac{2i+1}{2}\e \| \nabla U(x_0) \| + \e\sum_{j=1}^i \|\nabla U(x_{j\e}) - \nabla U(x_0)\| - \|p_0\|
\end{equation}
If for any $\delta>0$ we can choose an $x_0$ with large enough norm that $\sum_{j=1}^{i}\|x_{j\e} - x_{(j-1)\e}\| < \delta/2$ then $\sum_{j=1}^i \|\nabla U(x_{j\e}) - \nabla U(x_0)\|$ can be made arbitrarily small through the same continuity argument made above.  Thus, under $B$ it holds that $\|p_{(i+1/2)\e}\| \geq i\e \| \nabla U(x_0) \|$, from which it follows that $\|x_{(i+1)\e} - x_{i\e} \|$ can be made arbitrarily small by choosing $\|x_0\|$ large enough.
\end{proof}

\begin{proof}[Proof of proposition \ref{prop:lyapunov}]
It is shown in chapter 16 of \cite{meyn1993markov} that a geometric convergence bound is equivalent to the drift condition $\int V(y)P(x,dy) \leq \lambda V(x)$ whenever $x$ is outside some small set $C$, where $\lambda <1$.  Lemma B\ref{lem:compact} establishes that if \eqref{eqn:negmove} holds then any small set must be bounded.  Hence if a geometric bound holds here then
\begin{equation} \label{eqn:limbound}
\limsup_{\|x\| \to \infty} \frac{\int V(y)P(x,dy)}{V(x)} < 1.
\end{equation}
For any $\delta>0$ we can write
\begin{align*}
\int V(y)P(x,dy) &= \int_{B_\delta(x)} V(y)P(x,dy) + \int_{B_\delta^c(x)}V(y)P(x,dy), \\
&\geq \int_{B_\delta(x)} V(y)P(x,dy) + \epsilon,
\end{align*}
where $\epsilon = P\{x,B_\delta^c(x)\}$.  If (i) holds then we can choose a $\delta<\delta'$, so that 
\begin{align*}
\int_{\mathcal{B}_{\delta}(x)}e^{\log V(y)-\log V(x)}P(x,dy)+\epsilon  \geq\int_{\mathcal{B}_{\delta}(x)}e^{-\epsilon'}P(x,dy)+\epsilon 
  =e^{-\epsilon'}(1-\epsilon)+\epsilon.
\end{align*}
Noting that both $\epsilon$ and $\epsilon'$
can be made arbitrarily small as $\|x\|\to\infty$, this expression
tends to $1$ in the same limit, proving the result.  If (ii) holds, note that $\liminf_{\|x\|\to\infty}V(x)e^{-s\|x\|}=c$
implies that $\forall \epsilon'>0$ there is an $M<\infty$ such that $V(x)e^{-s\|x\|}\geq c-\epsilon'$
whenever $\|x\|\geq M$. This means that when $\|x\|>M$
\[
\int_{\mathcal{B}_{\delta}(x)}V(y)P(x,dy)+\epsilon\geq(c-\epsilon')\int_{\mathcal{B}_{\delta}(x)} e^{-s\|y\|}P(x,dy)+\epsilon.
\]
Condition (ii) also implies that for all $\epsilon'>0$, there is
a sequence $\{x_{i}\}_{i \geq 1}$ for which $\|x_{i}\|\to\infty$
as $i\to\infty$ such that whenever $i\geq N$ for some $N<\infty$
then $\|x_{i}\|>M$ and the condition $V(x_{i})e^{-s\|x_{i}\|}\leq c+\epsilon'$
holds. Combining gives that for all $i\geq N$
\[
\int\frac{V(y)}{V(x_{i})}P(x_{i},dy)\geq  \frac{(c-\epsilon')}{(c+\epsilon')} e^{-s\delta}(1-\epsilon)+\frac{\epsilon}{V(x_i)}.
\]
Since $\epsilon,\epsilon'$ and $\delta$ can all be made arbitrarily small and $V(x_{i})\to\infty$
as $\|x_{i}\|\to\infty$, then this proves the result.
\end{proof}

\begin{proof}[Proof of proposition \ref{prop:period}]
Assume $H(x_{0},p_{0})=E$ and $x_{0}=0$, $p_{0}=\left(\beta E\right)^{\frac{1}{\beta}}$.
Take $4T$ to be the period length, and note that by the symmetry of the Hamiltonian in question this implies that $p_{T}=0$ and $x_{T}=\left(\alpha E\right)^{\frac{1}{\alpha}}$.
Then 
\[
\mathcal{P}(E)=4\int_{0}^{T}dt = 4\int_{0}^{x_{T}}\frac{dt}{dx_{t}}dx_{t}=4\int_{0}^{x_{T}}p_{t}^{1-\beta}dx_{t}.
\]
Setting $b=(1-\beta)/\beta$, $c_{\beta}=\beta^{b}$ and noting
that $p_{t}^{1-\beta}=c_{\beta}(E-\alpha^{-1}x_{t}^{\alpha})^{b}$
for $t\in[0,T]$, then the expression can be written 
\[
\mathcal{P}(E)=4c_{\beta}\int_{0}^{x_{T}}\left(E-\alpha^{-1}x_{t}^{\alpha}\right)^{b}dx_{t}.
\]
Applying the change of variables $y_{t}=(\alpha E)^{-1/\alpha}x_{t}$
and setting $c_{\alpha}=\alpha^{1/\alpha}$ gives 
\begin{align*}
P(E)=4c_{\beta}c_{\alpha}E^{b+1/\alpha}\int_{0}^{1}(1-y_{t}^{\alpha})^{b}dy_{t},
\end{align*}
where . Now, we have that $\mathcal{P}(E)=f(E^{\eta})$, for some function $f$, where 
\[
\eta=\frac{1-\beta}{\beta}+\frac{1}{\alpha}=\frac{1-(\beta-1)(\alpha-1)}{\alpha\beta}.
\]
Setting $\alpha=1+\gamma$ and $\beta=1+\gamma^{-1}$ for some $\gamma>0$
gives 
\[
\eta=\frac{1-\gamma\gamma^{-1}}{(1+\gamma)(1+\gamma^{-1})}=0,
\]
as required.
\end{proof}

\begin{proof}[Proof of proposition \ref{prop:robust}]
Set $\gamma(x)=\min\left[ \frac{\e}{4}\|\nabla U(x)\|, \left\|\nabla^2 K\left\{\frac{\e}{4}\nabla U(x)\right\}\right\|^{-1/2} \right]$, and note that $\lim_{\|x\| \to \infty}\text{pr}_\nu\left\{\|p\|\leq \gamma(x)\right\}=1$.  For $\|p\|\leq \gamma(x)$, as a direct consequence of the mean value inequality \citep{dieudonne1961foundations}
\[
\left\| \nabla K\left\{ \frac{\e}{2}\nabla U(x) - p \right\} - \nabla K\left\{\frac{\e}{2}\nabla U(x)\right\} \right\| \leq M(x)\|p\|,
\]
where $M(x) = \sup_{\{4\|p\|\geq \e\|\nabla U(x)\|\}}\|\nabla^2 K(p)\|$. As the right-hand side tends to $0$ as $\|x\| \to \infty$, then the result follows.
\end{proof}

\section{Technical results}
\label{app:tech}

\begin{lemma} \label{lem:ht1}
If $\pi(\cdot)$ is heavy-tailed then for every $\gamma >0$
\[
\int e^{\gamma\|x\|-U(x)}dx = \infty.
\]
\end{lemma}

\begin{proof}
Choose $\delta<\gamma$. Let $B$ be a Euclidean ball centred at the origin such that $\|\nabla U(x)\|\leq \delta$ whenever $x \not\in B$. By continuity of $U(x)$, there is an $M<\infty$ such that $U(x) \leq M$ for all $x \in \partial B$. Then for all $x \not\in B$ the integrand is bounded below by $e^{(\gamma - \delta)\|x\| - M}$, which diverges uniformly and hence is not integrable.
\end{proof}

\begin{lemma} \label{lem:ht2}
If $\pi(x)$ is heavy-tailed then for any $\eta>0$ there is an $r<\infty$ such that 
\[
P\{x,B_{r}(x)\}>1-\eta.
\]
\end{lemma}

\begin{proof}
We need to show that $Q\{x,B_{r}(x)\}>1-\eta$,
for any $x$. After one leapfrog step we have 
\begin{align*}
x_{\e} & =x_{0}+\e\nabla K\left\{p_{0}-\frac{\e}{2}\nabla U(x_{0})\right\},\\
p_{\e} & =p_{0}-\frac{\e}{2}\nabla U(x_{0})-\frac{\e}{2}\nabla U(x_{\e}).
\end{align*}
Write $\|x\|_{\infty}$ for the supremum norm, and note that by equivalence
of norms in finite dimensions we can write $\|x\|_{\infty}\leq C\|x\|$
for all $x$, for some $C<\infty$. We have that $\nabla U(x)\in C_{0}(\R^{d})$,
which implies $\|\nabla U(x)\|<M/C$ for some $M<\infty$ which does
not depend on $x$, so that $\|\nabla U(x)\|_{\infty}<M$. The class of distributions for $\{p_0 - \e \nabla U(x_0)/2 \}$ is therefore tight. Now recall that if $f$ is a locally bounded function, and $\mathcal{F}$ a tight family of probability measures, then the resulting family of probability measures induced by pushing forward each element of $\mathcal{F}$ through $f$ is also tight.  So since $\nabla K$ is continuous and hence locally bounded, the result follows.
\end{proof}

\begin{lemma}\label{lem:kgrowth}
If \eqref{eq:cond1} and \eqref{eq:cond4} hold then 
\begin{equation}
\lim_{\|x\|\to\infty}\frac{\|\nabla K\{\frac{\varepsilon}{4}\nabla U(x)\}\|}{\|x\|}=\infty.
\label{eq:explode}
\end{equation}
\end{lemma}

\begin{proof} First we re-write the expression
\begin{equation*}
\lim_{\|x\|\to\infty}\frac{\|\nabla K\{\frac{\varepsilon}{4}\nabla U(x)\}\|}{\|x\|}
= \lim_{\|x\|\to\infty}\frac{\|\nabla K\{\frac{\varepsilon}{4}\nabla U(x)\}\|}{\|\nabla K\circ \nabla U(x)\|} \frac{\|\nabla K\circ \nabla U(x)\|}{\|x\|},
\end{equation*}
Now, \eqref{eq:cond4} implies that the first term will be bounded below by a finite positive constant, while \eqref{eq:cond1} ensures that the second will have an infinite limit, proving the result.
\end{proof}

\begin{lemma}
\label{lem:1stepbound} If $\pi(\cdot)$ is light-tailed, \eqref{eq:cond1} and either of \eqref{eq:cond5a} or \eqref{eq:cond5b} hold and $\|p_{0}\|\leq \frac{\varepsilon}{4}\min\{\|\nabla U(x_{0})\|,\|\nabla U(x_0)\|_\infty\}$ then there is a $\gamma_M < \infty$ such that, provided $\|x_{0}\|\geq\gamma_{M}$, it holds that $\|x_{\varepsilon}\|\geq M\|x_{0}\|$,
for any $M<\infty$.
\end{lemma}
\begin{proof} Note
\begin{align*}
\|x_\e\| = \|x_0 + \e\nabla K \left\{ p_0 - \frac{\e}{2}\nabla U(x_0) \right\} \|
\geq \e\|\nabla K \left\{ p_0 - \frac{\e}{2}\nabla U(x_0) \right\} \| - \|x_0\|.
\end{align*}
It is therefore sufficient to show that for any $M<\infty$ we can choose an $\|x_0\|$ large enough that
\[
\|\nabla K \left\{ p_0 - \frac{\e}{2}\nabla U(x_0) \right\} \| \geq \frac{(M+1)}{\e}\|x_0\|.
\]
Under \eqref{eq:cond5a}, note that 
\begin{align*}
\|p_0 - \frac{\e}{2}\nabla U(x_0)\| \geq \frac{\e}{2}\|\nabla U(x_0)\| - \|p_0\| \geq \frac{\e}{4}\|\nabla U(x_0)\|,
\end{align*}
which implies
\[
\|\nabla K \left\{ p_0 - \frac{\e}{2}\nabla U(x_0) \right\} \| \geq \|\nabla K\left\{ \frac{\e}{4}\nabla U(x_0) \right\} \|.
\]
By \eqref{eq:explode}, therefore, if $\|x_0\|$ is chosen to be large enough then this can be made $\geq (M+1)\|x_0\|/\e$, for any finite $M$, proving the result.

Under \eqref{eq:cond5b}, recall that there exists global constants $C,c>0$ such that $C\|\nabla U(x)\| \geq \|\nabla U(x)\|_\infty \geq c\|\nabla U(x)\|$ for all $x \in \mathbb{R}^d$. It suffices in this setting therefore to show that we can choose an $\|x_0\|$ large enough that
\[
\|\nabla K \left\{ p_0 - \frac{\e}{2}\nabla U(x_0) \right\} \|_\infty \geq \frac{C(M+1)}{\e}\|x_0\|.
\]
We have
\[
\|\nabla K \left\{ p_0 - \frac{\e}{2}\nabla U(x_0) \right\} \|_\infty= \max_j |k' \{p_0(j) - \partial_j U(x_0)\}|.
\]
Write $i^*$ and $j^*$ to denote the indices for which $\|p_0-\nabla U(x_0)\|_\infty = |p_0(i^*) - \partial_{i^*}U(x_0)|$ and $\|\nabla U(x_0)\|_\infty = |\partial_{j^*} U(x_0)|$. We have:
\begin{align*}
\| p_0 - \frac{\e}{2}\nabla U(x_0)  \|_\infty &= |p_0(i^*) - \partial_{i^*}U(x_0)| \\
&\geq |p_0(j^*) - \partial_{j^*} U(x_0)| \\
&\geq \frac{\e}{2}|\partial_{j^*} U(x_0)| - |p_0(j^*)| \\
&\geq \frac{\e}{4}|\partial_{j^*} U(x_0)|.
\end{align*}
Now, $|p_0(i^*) - (\e/2)\partial_{i^*}U(x_0)| \geq (\e/4)|\partial_{j^*}U(x_0)|$ implies that $|k'\{p_0(i^*) - (\e/2)\partial_{i^*}U(x_0) \}| \geq |k'\{(\e/4)\partial_{j^*}U(x_0)\}| = \|\nabla K\{(\e/4)\nabla U(x_0)\}\|_\infty$. Using the global bounds then we see that for any $M<\infty$ we can choose an $\|x_0\|$ large enough that
\[
\frac{\|\nabla K\{(\e/4)\nabla U(x_0)\}\|_\infty}{\|x_0\|} \geq \frac{C(M+1)}{\e},
\]
establishing the result.
\end{proof}

\begin{lemma}\label{lem:Lstepbound} If $\pi(\cdot)$ is light-tailed and \eqref{eq:cond1}-\eqref{eq:cond4} and one of \eqref{eq:cond5a} and \eqref{eq:cond5b} hold, and provided that for any fixed $i\geq 0$\\
 (i)  $\|x_0\| \geq \gamma_M$ for some $\gamma_M <\infty$, \\
 (ii)  $\|p_0\| \leq (\e/4)\min\{\|\nabla U(x_0)\|, \|\nabla U(x_0)\|_\infty \}$, \\
 (iii)  $M$ is large enough that $\phi(M)\geq 7/3$ with $\phi$ as in \eqref{eq:cond2}, \\
 (iv)  $\|x_{j\e}\|\geq M\|x_{(j-1)\e}\|$ for all $j \leq i$, \\
it holds for any finite $M<\infty$ that $\|x_{(i+1)\e}\| \geq M\|x_{i\e}\|$.
\end{lemma}


\begin{proof} We first show
\begin{equation} \label{eq:b5pbound}
\|p_{(i+\frac{1}{2})\e}\| \geq \frac{\e}{4}\|\nabla U(x_{i\e})\|.
\end{equation}
To show \eqref{eq:b5pbound}, first note by iterating \eqref{eqn:momentum_update} and noting $p_{(i+1/2)\e} = p_{i\e} - \e\nabla U(x_{i\e})/2$ that
\begin{align*}
\|p_{(i+\frac{1}{2})\e}\| &= \|p_0 - \frac{\e}{2}\nabla U(x_0) - \e \sum_{j=1}^i\nabla U(x_{j\e}) \| \\
 &\geq \e\|\nabla U(x_{i\e})\| - \frac{\e}{2}\|\nabla U(x_0)\| - \e \sum_{j=1}^{i-1} \|\nabla U(x_{j\e})\| - \|p_0\|.
\end{align*}
Using the stated assumption that $\|p_0\| \leq (\e/4)\|\nabla U(x_0)\|$ then gives
\[
\|p_{(i+\frac{1}{2})\e}\| \geq \e \left\{ \|\nabla U(x_{i\e})\| - \sum_{j=0}^{i-1} \|\nabla U(x_{j\e})\| \right\}. 
\]
Now, \eqref{eq:cond2} implies that if $\|x_{(j-1)\e}\| \leq M^{-1} \|x_{j\e}\|$, $\|\nabla U(x_{(j-1)\e})\| \leq \phi(M)^{-1}\|\nabla U(x_{j\e})\|$. Substituting into the above expression gives
\[
\|p_{(i+\frac{1}{2})\e}\| \geq \e \left\{ 1 - \sum_{j=0}^{i-1} \phi(M)^{j-i} \right\}\|\nabla U(x_{i\e})\|.
\]
To finish the argument we need to therefore show that $\{ 1 - \sum_{j=0}^{i-1} \phi(M)^{j-i} \} \geq 1/4$. First note that
\[
\sum_{j=0}^{i-1} \phi(M)^{j-i} = \phi(M)^{-i} + \phi(M)^{-(i-1)} + ... + \phi(M)^{-1} = \sum_{j=1}^i \phi(M)^{-j}.
\]
Using simple geometric series identities gives 
\begin{align*}
\sum_{j=1}^i \phi(M)^{-j} &= \frac{1-\phi(M)^{-(i+1)}}{1-\phi(M)^{-1}} - 1 \\
&= \frac{\phi(M) - \phi(M)^{-i}}{\phi(M) - 1} - \frac{\phi(M) - 1}{\phi(M) - 1} \\
&= \frac{1-\phi(M)^{-i}}{\phi(M) - 1}.
\end{align*}
Hence $\{ 1 - \sum_{j=0}^{i-1} \phi(M)^{j-i} \} \geq 1/4$ if 
\[
1 - \frac{1-\phi(M)^{-i}}{\phi(M) - 1} \geq 1/4,
\]
which will hold if $\phi(M) \geq 7/3$ for any $i\geq 1$. The stated assumption that $\phi(M)\geq 7/3$ therefore establishes \eqref{eq:b5pbound}.

If \eqref{eq:cond5a} holds, then \eqref{eq:b5pbound} can be combined with \eqref{eq:explode} directly to show that for any $M<\infty$ and $\e \in (0,\infty)$, if $\|x_{i\e}\|\geq \gamma_M$ for some suitably chosen $\gamma_M <\infty$ it will hold that
\[
\|\nabla K\{p_{(i+\frac{1}{2})\e}\}\| \geq \|\nabla K\left\{ \frac{\e}{4}\nabla U(x_{i\e})\right\}\| \geq \frac{(M+1)}{\e}\|x_{i\e}\|.
\]
If \eqref{eq:cond5b} holds instead of \eqref{eq:cond5a}, then the same result can be established using a similar argument to that given in the last paragraph of the proof of Lemma B\ref{lem:1stepbound}. Hence, provided that $\|x_0\| \geq \gamma_M$, for all $i\geq 0$ it holds that 
\begin{align*} 
\|x_{(i+1)\e}\| &= \|x_{i\e} + \e\nabla K\{p_{(i+\frac{1}{2})\e}\}\| \\
 &\geq \e\|\nabla K\{p_{(i+\frac{1}{2})\e}\}\|- \|x_{i\e}\| \\
 &\geq M\|x_{i\e}\|,
\end{align*}
which proves the result.

\end{proof}

\begin{lemma}\label{lem:pbound} Under the same conditions as Lemma B\ref{lem:Lstepbound} and provided $M$ is such that $\phi(M)\geq 5$, then $\|p_{L\varepsilon}\| \geq \|p_{0}\|$
for any $L \geq 1$.
\end{lemma}
\begin{proof} We have
\begin{align*}
\|p_{L\e}\| &= \|p_0 - \frac{\e}{2}\left\{ \nabla U(x_0) + \nabla U(x_{L\e})\right\} - \e\sum_{i=1}^{L-1}\nabla U(x_{i\e}) \| \\
&\geq \frac{\e}{2} \left\{ \|\nabla U(x_{L\e})\|-2\sum_{i=0}^{L-1}\|\nabla U(x_{i\e})\| \right\},
\end{align*}
by recalling that as in Lemma B\ref{lem:Lstepbound} $\|p_0\| \leq (\e/4)\|\nabla U(x_0)\|$. Using \eqref{eq:cond2} and the stated assumptions we have for any $i \leq L$  that
\[
\|\nabla U(x_{i\e})\| \leq \phi(M)^{i-L}\|\nabla U(x_{L\e})\|,
\]
which implies
\[
\|p_{L\e}\| \geq \frac{\e}{2}\left\{ 1-2\sum_{i=0}^{L-1}\phi(M)^{i-L} \right\}\|\nabla U(x_{L\e})\|
\]
The stated assumptions $\|p_0\| \leq (\e/4)\|\nabla U(x_0)\|$ and $\phi(M)^L\|\nabla U(x_0)\| \leq \|\nabla U(x_{L\e})\|$ lead to the bound
\[
\|\nabla U(x_{L\e})\| \geq \phi(M)^L \frac{4}{\e} \|p_0\|,
\]
which, when combined with the above inequality, give
\[
\|p_{L\e}\| \geq 2\left\{ 1-2\sum_{i=0}^{L-1}\phi(M)^{i-L} \right\}\phi(M)^L \|p_0\|.
\]
As shown in the proof of Lemma B\ref{lem:Lstepbound} $\sum_{i=0}^{L-1}\phi(M)^{i-L} = \{1-\phi(M)^{-L}\}/\{\phi(M) - 1\}$. Hence the result is proven if
\[
2\left\{ 1-2\frac{1-\phi(M)^{-L}}{\phi(M)-1} \right\}\phi(M)^L \geq 1,
\]
which will indeed be true under the stated assumption that $\phi(M) >5$.
\end{proof}

\begin{lemma}
\label{lem:compact} If \eqref{eqn:negmove} holds for a Hamiltonian Monte Carlo method then any small set must be bounded.
\end{lemma}
\begin{proof} Since $\nabla K(p)$ and $\nabla U(x)$ are continuous and $\nabla K \circ \nabla U(x)$ is vanishing at infinity, then $\nabla K \circ \nabla U(x)$ is also bounded, implying that the collection of Hamiltonian Monte Carlo increments $\{ P(x,\cdot) - x \}$ is uniformly tight, using a similar argument to that employed in the proof of Lemma B\ref{lem:ht2}.  Hence, any small set must be bounded, following Lemma 2.2 of \cite{jarner2000geometric}.
\end{proof}

\section{Numerical example of the efficiency-robustness trade-off}

We consider the distribution $\pi(x) \propto \exp\left( - \beta^{-1}|x|^\beta \right)$ with $\beta = 1.5$, and compare the quadratic and relativistic kinetic energy choices.  We test 80 evenly spaced choices of step-size $\e$ spanning from 0.1 to 5, and in each case begin the sampler at equilibrium and compute the empirical expected squared jump distance from a chain of length 200,000, with the number of leapfrog steps randomly selected uniformly between 1 and 5 at each iteration.  The results are shown in Figure \ref{figA}.
\begin{figure}
\centering
\includegraphics[height=15pc]{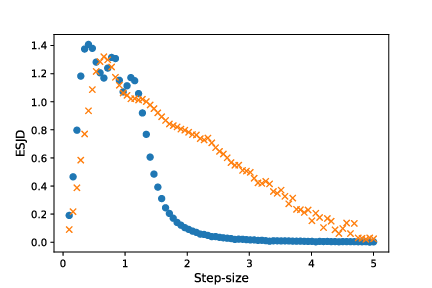}
\caption{Empirical expected squared jump distance versus step-size for the relativistic kinetic energy, denoted by orange crosses, and quadratic kinetic energy, denoted by blue dots.}
\label{figA}
\end{figure}
As can be seen, the quadratic choice leads to a higher optimal value, but when step-sizes are chosen to be too large the jump distance drops quickly. The relativistic choice, by comparison, exhibits a larger degree of robustness to bigger step-sizes.

\end{document}